\documentclass{article}
\usepackage{a4wide}
\usepackage[english]{babel}
\usepackage{graphicx}
\usepackage{amsmath,amsfonts}
\usepackage{amssymb,amsthm}
\usepackage{xspace}
\usepackage{xcolor}
\usepackage{microtype}%

\usepackage{xspace}

\theoremstyle{plain}

\graphicspath{{./images/}}%

\newtheorem{defin}{Definition}
 
\newtheorem{theo}[defin]{Theorem}
 \newenvironment{theorem}{\begin{theo} \sl}{\end{theo}}
\newtheorem{lem}[defin]{Lemma}
 \newenvironment{lemma}{\begin{lem} \sl}{\end{lem}}
\newtheorem{coro}[defin]{Corollary}
 
\newtheorem{prop}[defin]{Proposition}

\newcommand{\Vis}{\mathord{\it Vis}}

\newcommand{\etal}{\emph{et~al.}\xspace}

\title{Routing on the Visibility Graph
\footnote{
An extended abstract of this paper appeared in the proceedings of the 28th International Symposium on Algorithms and Computation (ISAAC 2017)~\cite{BKRV2017RoutingVisibilityGraph}. 
P.~B. is supported in part by NSERC. 
M.~K.~partially supported by MEXT KAKENHI No.~17K12635 and the NSF award CCF-1422311.
M.~K. and A.~v.~R. was supported by JST ERATO Grant Number JPMJER1201, Japan. 
S.~V.~is supported in part by NSERC and the Carleton-Fields Postdoctoral Award.}}

\author{Prosenjit Bose\thanks{Carleton University, Ottawa, Canada. {\tt \{jit,sander\}@scs.carleton.ca}}
  \and Matias Korman\thanks{Tufts University, Medford, USA. {\tt matias.korman@tufts.edu}} 
  \and Andr\'e van Renssen\thanks{The University of Sydney, Sydney, Australia. {\tt andre.vanrenssen@sydney.edu.au}}  
  \and Sander Verdonschot$^\ddagger$}

\date{}

\begin{document}

\maketitle

\begin{abstract}
  We consider the problem of routing on a network in the presence of line segment constraints (i.e., obstacles that edges in our network are not allowed to cross). Let $P$ be a set of $n$ points in the plane and let $S$ be a set of non-crossing line segments whose endpoints are in $P$. We present two deterministic 1-local $O(1)$-memory routing algorithms (i.e., the algorithms never look beyond the direct neighbours of the current location and store only a constant amount of additional information). These algorithms are guaranteed to find a path consisting of at most a linear number of edges between any pair of vertices of the \emph{visibility graph} of $P$ subject to a set of constraints $S$. Contrary to {\em all} existing deterministic local routing algorithms, our routing algorithms do not route on a plane subgraph of the visibility graph. Additionally, we provide lower bounds on the routing ratio of any  deterministic local routing algorithm on the visibility graph. 
\end{abstract}

\section{Introduction}
Routing is the process of sending a message through a network from a source vertex to a destination vertex. It is a fundamental problem in networks. If the routing algorithm has complete knowledge of the network, it is known how to send a message from any source vertex to any destination vertex (among others,  Dijkstra's algorithm can compute a shortest path between two vertices in a network).

However, it is not always possible for the routing algorithm to have full knowledge of the network. If the network is very large it may be too expensive to store it explicitly. And even if storage constraints are not an issue, it may be hard to keep the representation up-to-date if the network changes frequently.

Thus, there is a need for routing algorithms that guarantee message delivery and use as little information of the network as possible. For example, it is challenging to successfully route when the only information available to the routing algorithm is the location of the current vertex, the neighbours of the current vertex and a constant amount of additional information, such as the original source vertex of the message and the destination vertex. A routing algorithm that can work under these constraints is often referred to as {\em local} (or $k$-local for some constant $k$, when the $k$-neighbourhood\footnote{The $k$-neighbourhood is the set of vertices reachable in at most $k$ steps from the current vertex.} is available). In our setting, we assume that the network is a graph embedded in the plane, with edges as straight line segments connecting pairs of vertices, weighted by the Euclidean distance between their endpoints. We refer to such networks as {\em geometric networks}. Algorithms routing on such networks are referred to as \emph{geometric} routing algorithms \mbox{(see \cite{G09} and \cite{R09}} for surveys of the area). 

Since local routing algorithms have little information beyond the immediate neighborhood of the current vertex, they use the structure of the network to guide their navigation. For example, intuitively speaking, the $\Theta_m$-graph is a geometric network where each point connects to its nearest point in $m$ different cones, which can be thought of as directions (a formal definition is given in Section \ref{sec:prelim}). Thus, to route on this graph, it has been shown \cite{KG92} that it suffices to send the message to the nearest node in the direction of the desired destination. If $m$ is seven or larger, this strategy will ensure that the message always reaches its destination and its length is not much more than the Euclidean distance between the source and destination (see \cite{BCMRS16} for a survey on $\Theta_m$-graphs). The argument can be extended to show that when $m$ is six, the strategy always reaches the destination, but this no longer gives a guarantee on the length of the path. 

Using the structure of a graph to guide a routing strategy becomes more challenging in the presence of constraints, since constraints act as barriers and can disrupt the inherent structure that may be present in graphs built without constraints. For example a vertex in the $\Theta_m$-graph may no longer have an edge in a given direction because of existing constraints. We model this more general setting in which some connections are forbidden by using a set $S$ of non-intersecting {\em line segment constraints} whose endpoints are vertices of the network. These segments act as constraints in the sense that no edge can properly intersect an edge of $S$.

Given a set $P$ of $n$ points in the plane and a set $S$ of non-intersecting line segment constraints, we say that two vertices $u$ and $v$ can \textit{see each other} (or are {\em visible}) when either the line segment $uv$ does not properly intersect any constraint in $S$ or $uv$ is itself a constraint in $S$. If two vertices $u$ and $v$ can see each other, the line segment $uv$ is referred to as a \emph{visibility edge}. The \emph{visibility graph} of $P$ with respect to a set of constraints $S$, denoted $\Vis(P,S)$, has $P$ as vertex set and all visibility edges as edge set. 

This setting has been studied extensively in the context of motion planning amid obstacles. Clarkson \cite{C87} was one of the first to study this problem. In his work, he showed how to find an approximate shortest path between two points in the plane amid a set of obstacles.  In order to find this approximate shortest path efficiently, Clarkson constructs a $(1+\epsilon)$-spanner of $\Vis(P,S)$ with a linear number of edges. A subgraph $H$ of $G$ is called a $t$-spanner of $G$ (for $t\geq 1$) if for each pair of vertices $u$ and $v$, the shortest path in $H$ between $u$ and $v$ has length at most $t$ times the shortest path between $u$ and $v$ in $G$. The smallest value $t$ for which $H$ is a $t$-spanner is the \emph{spanning ratio} or \emph{stretch factor} of $H$. Following Clarkson's result, Das \cite{D97} showed how to construct a spanner of $\Vis(P,S)$ with constant spanning ratio and constant degree. Bose and Keil \cite{BK06} showed that the Constrained Delaunay Triangulation is a 2.42-spanner of $\Vis(P,S)$. Recently, the constrained half-$\Theta_6$-graph (which is identical to the constrained Delaunay graph whose empty visible region is an equilateral triangle) was shown to be a plane 2-spanner of $\Vis(P,S)$~\cite{BFRV12Constrained} and all constrained $\Theta$-graphs with at least 6 cones were shown to be spanners as well~\cite{BR14}. 

To the best of our knowledge, all deterministic routing algorithms known to date that guarantee that each message eventually reaches its intended destination in geometric networks compute some plane subgraph of the complete Euclidean graph and somehow route on the subgraph. This means that of the potentially quadratic number of edges available to the routing algorithm, only a linear number are ever considered. This artificial constraint limits the number of options available and thus can create paths that are much longer than necessary, for example when the destination vertex is visible from the source vertex. The visibility graph $\Vis(P,S)$ depicts all connections that are not blocked by the set $S$ of constraints. In other words, it has all connections that can be used (and this graph need not be plane). In this paper we present a strategy to route on the visibility graph, making it the first deterministic local routing algorithm that does not restrict its choices to a plane subgraph of $\Vis(P,S)$.

\subsection{Results and previous work}
Although motion planning amid obstacles has been studied extensively~\cite{C87, D97, BFRV12Constrained, BR14}, there has not been much work on routing in the same setting. Bose~\etal~\cite{BFRV2017RoutingJournal} showed that it is possible to route locally and 2-competitively between any two visible vertices in the constrained $\Theta_6$-graph (the constrained $\Theta_m$-graph with 6 cones). A routing strategy is called $c$-competitive when the length of the path that the routing strategy follows is at most $c$ times the length of the shortest path between the source and destination in the graph. Additionally, an 18-competitive routing algorithm between any two visible vertices in the constrained half-$\Theta_6$-graph (which is equivalent to the constrained Delaunay graph that uses an empty equilateral triangle) was provided~\cite{BFRV2017RoutingJournal}, but this strategy does not work when the source and destination do not see each other. In the same paper it was shown that no deterministic local routing algorithm is $o(\sqrt{n})$-competitive between all pairs of vertices of the constrained $\Theta_6$-graph, regardless of the amount of memory it is allowed to use. Recently, the authors presented a non-competitive 1-local $O(1)$-memory routing algorithm to route on the visibility graph~\cite{BKRV2017Routing} (a formal definition of such an algorithm can be found in Section~\ref{sec:prelim}). However, this method also restricts the edge choices to a plane subgraph of $\Vis(P,S)$, as it locally determines the edges of the so-called constrained half-$\Theta_6$-graph and routes on it.

We present two deterministic 1-local $O(1)$-memory routing algorithms on $\Vis(P,S)$. The first algorithm locally computes a non-plane subgraph of the visibility graph (the constrained $\Theta_6$-graph) and routes on it. We then modify this algorithm to obtain a routing algorithm that routes directly on the visibility graph (i.e., any edge of $\Vis(P,S)$ may be used). Both of these algorithms reach the destination in at most $O(n)$ steps. To the best of our knowledge, this is the first local routing algorithm that does not compute a plane subgraph of the visibility graph. 

We also provide some lower bounds to the problem. Specifically, we show that no deterministic local routing algorithm can be $o(n)$-competitive if we measure the quality of the path by the number of steps taken between every pair of vertices in the visibility graph. Our routing strategy creates paths of linear length, hence they are the best we can hope for in this regard. 

Alternatively, if we measure the quality of the path with respect to the Euclidean length of the path, we show that no algorithm can be $o(\sqrt{n})$-competitive. Under specific conditions we can also show that no algorithm can be $o(n)$-competitive. This second bound only holds if the algorithm considers only the subgraph induced by the endpoints of edges crossing the line segment between the source and destination (a technique commonly used in the unconstrained setting). 

\section{Preliminaries}\label{sec:prelim}
The $\Theta_m$-graph plays an important role in our routing strategy. We begin by defining it. Define a \emph{cone} $C$ to be the region in the plane between two rays originating from a vertex referred to as the \emph{apex} of the cone. When constructing a (constrained) $\Theta_m$-graph, for each vertex $u$ consider the rays originating from $u$ with the angle between consecutive rays being $2 \pi / m$. Each pair of consecutive rays defines a cone. The cones are oriented such that the bisector of some cone coincides with the vertical ray emanating from $u$ that lies above $u$. Let this cone be $C_0$ of $u$ and number the cones in clockwise order around $u$ (see Figure~\ref{fig:Cones}). The cones around the other vertices have the same orientation as the ones around $u$. We write $C_i^u$ to indicate the $i$-th cone of a vertex $u$, or $C_i$ if $u$ is clear from the context. For ease of exposition, we only consider point sets in general position: no two points lie on a line parallel to one of the rays that define the cones, no two points lie on a line perpendicular to the bisector of a cone, and no three points are collinear. The main implication of this assumption is that no point lies on a cone boundary. These assumptions can be removed via classic symbolic perturbation techniques.

\begin{figure}[ht]
  \begin{minipage}[t]{0.45\linewidth}
    \begin{center}
      \includegraphics{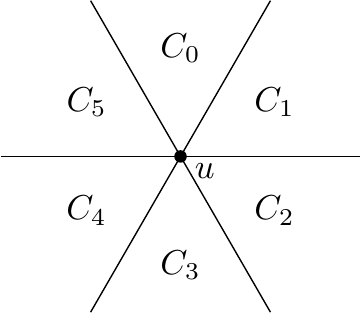}
    \end{center}
    \caption{The cones with apex $u$ in the $\Theta_6$-graph. All points of $S$ have exactly six cones.}
    \label{fig:Cones}
  \end{minipage}
  \hspace{0.05\linewidth}
  \begin{minipage}[t]{0.45\linewidth}
    \begin{center}
      \includegraphics{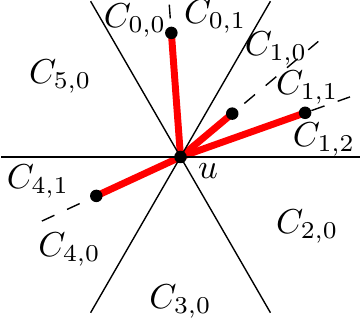}
    \end{center}
    \caption{The subcones with apex $u$ in the constrained $\Theta_6$-graph (constraints denoted as red thick segments).}     \label{fig:ConstrainedCones}
  \end{minipage}
\end{figure}

Let vertex $u$ be an endpoint of a constraint $c$ (if any) and let $v$ be the other endpoint and let cone $C_i^u$ be the cone that contains $v$. The lines through all constraints $c$ in $C_i^u$ with $u$ as an endpoint split $C_i^u$ into several \emph{subcones} (see Figure~\ref{fig:ConstrainedCones}). We use $C_{i, j}^u$ to denote the $j$-th subcone of $C_i^u$ (again, numbered in clockwise order). When a constraint $c = (u, v)$ splits a cone of $u$ into two subcones, we define $v$ to lie in both of these subcones. We consider a cone that is not split to be a single subcone.

We now introduce the {\em constrained} $\Theta_m$-graph: for each subcone $C_{i, j}$ of each vertex $u$, add an edge from $u$ to the closest vertex in that subcone that can see $u$, where distance is measured along the bisector of the original cone (\emph{not the subcone}). More formally, we add an edge between two vertices $u$ and $v$ if $v$ can see $u$, $v \in C_{i, j}^u$, and for all points $w \in C_{i, j}^u$ that can see $u$, $|u v'| \leq |u w'|$, where $v'$ and $w'$ denote the projection of $v$ and $w$ on the bisector of $C_i^u$ and $|x y|$ denotes the length of the line segment between two points $x$ and $y$. Note that our general position assumption implies that each vertex adds at most one edge per subcone.

We now define our routing model. Formally, a routing algorithm $A$ is a deterministic $1$-local, $O(1)$-memory routing algorithm, if the choice of the vertex to which a message is forwarded from the current vertex $s$ is a function of $s$, $t$, $N(s)$, and $M$, where $t$ is the destination vertex, $N(s)$ is the set of vertices adjacent to $s$ and set of constraints incident to $s$ and $M$ is a memory of constant size, stored with the message. We consider a unit of memory to consist of a $\log_2 n$ bit integer or a point in $P$. Our model assumes that the only information stored at each vertex of the graph is $N(s)$.

\begin{lemma}{\emph{\cite{BFRV12Constrained}}}
  \label{lem:ConvexChain}
Let $u$, $v$, and $w$ be three arbitrary points in the plane such that $u w$ and $v w$ are visibility edges and $w$ is not the endpoint of a constraint intersecting the interior of triangle $u v w$ (see Figure~\ref{fig:VisiblePointInsideTriangle}). Then there exists a convex chain of visibility edges from $u$ to $v$ in triangle $u v w$, such that the polygon defined by $u w$, $w v$ and the convex chain is empty and does not contain any constraints.
\end{lemma}

\begin{figure}[ht]
  \begin{center}
    \includegraphics{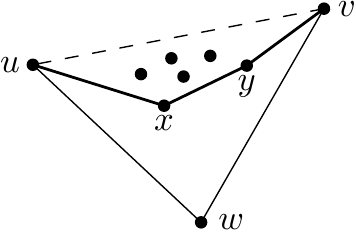}
  \end{center}
  \caption{A convex chain from $u$ to $v$ via $x$ and $y$.}
  \label{fig:VisiblePointInsideTriangle}
\end{figure}

If $u$ and $v$ do not see each other, the above lemma proves the existence of a convex path between them. We use this property repeatedly in our routing algorithm.

\section{Routing on the Constrained $\boldsymbol{\Theta_6}$-Graph}\label{sec:routing}

Prior to describing our routing strategy for the entire visibility graph, we first provide one for the constrained $\Theta_6$-graph. Note that the $\Theta_6$-graph is not necessarily plane. In this section, we assume that we are given the constrained $\Theta_6$-graph explicitly. In the next section, we show how to use this algorithm to route on the visibility graph. 

If there are no constraints, there exists a simple local routing algorithm that works on all $\Theta$-graphs with at least 4 cones. This routing algorithm, which we call $\Theta$-routing, always follows the edge to the closest vertex in the cone that contains the destination. In the constrained setting, this algorithm follows the edge to the closest vertex in the \emph{subcone} that contains the destination. Unfortunately, this approach does not necessarily succeed in the constrained setting due to two issues. First, a key factor of convergence in the unconstrained $\Theta$-routing algorithm is that each step gets us closer to the destination (as long as we have at least 6 cones). Unfortunately, this property need not hold in the constrained setting (see Figure~\ref{fig:ThetaRoutingLongStep}a). 

\begin{figure}[ht]
  \begin{center}
    \includegraphics{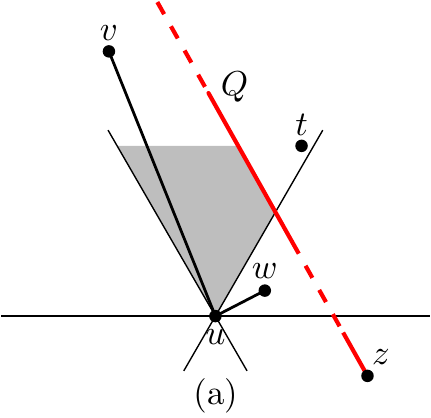}
    \hskip0.5cm
    \includegraphics{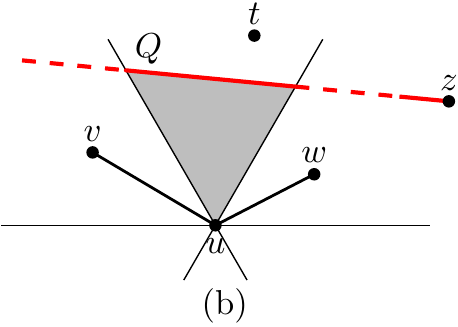}
    \hskip0.5cm
    \includegraphics{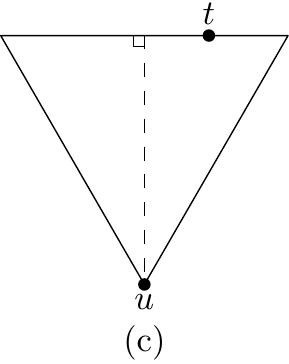}
  \end{center}
  \caption{(a) The situation in which $\Theta$-routing follows an edge to $v$ and ends up further away from the destination. (b) The situation where the $\Theta$-routing algorithm cannot follow any edges at $u$, since the destination $t$ lies behind a constraint. (c) The canonical triangle of $u$, $\triangle_{ut}$.}
  \label{fig:ThetaRoutingLongStep}
\end{figure}

A second, more important problem is that the cone containing the destination need not contain any visible vertices. This happens when a constraint is directly blocking visibility (see Figure~\ref{fig:ThetaRoutingLongStep}b). In this case, the $\Theta$-routing algorithm will get stuck, since it cannot follow any edge in that cone. 

The first problem can be easily fixed: given a vertex $u$ and the destination $t$, we define the \emph{canonical triangle} of $u$ with respect to $t$, denoted $\triangle_{ut}$, as the triangle with apex $u$, bounded by the cone boundaries of the cone of $u$ that contains $t$ and the line through $t$ perpendicular to the bisector of the cone (see Figure~\ref{fig:ThetaRoutingLongStep}c). If the edge of $u$ that lies in that cone ends outside the canonical triangle, we call the edge {\em invalid} and we ignore it. By ignoring invalid edges we make sure that any edge we follow leads to a vertex that is closer to $t$.  

To solve the second problem, the routing algorithm needs to find a path even when an obstacle is blocking visibility to the destination (either blocking all visibility from $u$ in the cone of $t$ or because the edge in that cone is invalid). In this case the algorithm enters the {\em obstacle avoidance phase}, routing differently until an endpoint of a blocking constraint is reached. 

Intuitively, our algorithm uses the $\Theta$-routing algorithm until it gets stuck, at which point it switches to the obstacle avoidance phase in order to get around a constraint blocking its visibility to $t$. After this phase ends, the algorithm switches back to the $\Theta$-routing algorithm. This process is repeated until $t$ is reached. A more precise description follows in Section~\ref{sec:GlobalStrategy}. 

\subsection{Obstacle Avoidance Phase}
We first describe the obstacle avoidance phase. The algorithm enters this phase when routing from source $s$ to destination $t$, and reaches a vertex $u$ that does not have any valid edges in the cone that contains $t$. This can only happen if a constraint $Q$ is blocking visibility to $t$ (if many of them exist, let $Q$ be the one whose intersection with segment $\overline{ut}$ is closest to $u$). The goal of this phase is to reach the right endpoint of $Q$, which we denote as $z$. The main difficulty with this phase is that the algorithm does not know where $z$ is, since $Q$ is not incident on $u$. In order to overcome this difficulty, the algorithm exploits several geometric properties arising from the unique symmetries present in the constrained $\Theta_6$-graph, some of which are outlined in the proof of Lemma \ref{lemma:Reaching the endpoint-1}.

Without loss of generality, $t$ lies in $C^u_0$. We first describe the case where $u$ has no edges in $C_0$. The general case, where $u$ may have invalid edges in $C_0$, will be considered afterwards. In this first case, the algorithm proceeds as follows. At a current vertex $m$,  the algorithm considers one of two candidate edges to follow (see Figure~\ref{fig:RoutingChoice}). The first is the edge to the closest visible vertex $v$ in the subcone of $C^m_2$ that shares a boundary with $C^m_1$. The second edge is the edge from $m$ to the vertex $w$ in $C^m_1$ that minimizes the angle $\alpha$ between $\overline{mw}$ and the right boundary of $C^m_0$. If $v$ lies in $C^w_4$ and $m$ is not the endpoint of a constraint that intersects the interior of triangle $m v w$, the algorithm follows the edge to $v$. Otherwise, it follows the edge to $w$. In the proof of Lemma \ref{lemma:Reaching the endpoint-1}, we show that at least one of $v$ or $w$ exists. If one of the two vertices $v$ or $w$ does not exist, the algorithm follows the edge that does exist.  The obstacle avoidance phase ends when the algorithm reaches the endpoint of a constraint that intersects $\overline{ut}$. In order to recognize this, the algorithm stores $u$ when the phase begins.

\begin{figure}[ht]
  \begin{center}
    \includegraphics{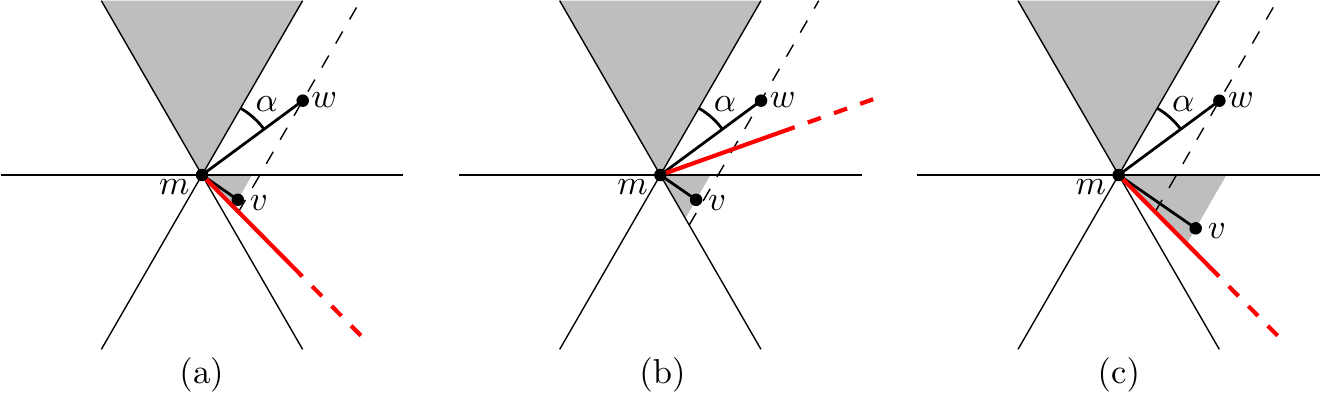}
  \end{center}
  \caption{Routing from a vertex $m$. (a) Follow the edge to $v$, since $v$ lies in $C^w_4$. (b) Follow the edge to $w$, since $m$ is the endpoint of a constraint that intersects $m v w$. (c)~Follow the edge to $w$, since $v$ lies outside of $C^w_4$.}
  \label{fig:RoutingChoice}
\end{figure}

\begin{lemma}\label{lemma:Reaching the endpoint-1}
When $u$ has no edges in the cone containing the destination $t$, the obstacle avoidance phase initiated by $u$ reaches the right endpoint $z$ of the closest constraint $Q$ blocking visibility to $t$.
\end{lemma}
\begin{proof}
Without loss of generality, let $t$ lie in $C^u_0$. Since $u$ has no edges in $C_0$, the closest constraint $Q$ must intersect both boundaries of $C_0^u$. This implies that $z$ is either in $C^u_1$ or $C^u_2$. We maintain the invariant that each intermediate vertex $m$ has no edges in $C_0^m$ and that the intersection of the right boundary of $C^m_0$ and $Q$ is closer to $z$ than in the previous step. We first show that there always exists either a $w$ in $C_1^m$ or a $v$ in $C_2^m$ as defined in the paragraph preceding this lemma. This implies that our algorithm eventually reaches $z$ since there are a finite number of points in $P$.

As a consequence of our invariant, $z$ must either lie in $C_1^m$ or $C_2^m$.  Since $m$ has no edges in $C_0$, we have that $Q$ is the closest constraint to $m$ in $C_0^m$. Thus, any point $x$ on $Q \cap C_0^m$ is visible from both $m$ and $z$. Hence, we can apply Lemma~\ref{lem:ConvexChain} to the triangle $m x z$ and obtain a convex chain of visibility edges from $m$ to $z$. In particular, this implies that $m$ can see a vertex in $C_1 \cup C_2$, and therefore it has an edge in $C_1 \cup C_2$. What remains to be shown is that the invariant is maintained after every step of the algorithm. We note that for any vertex in $C_1^m \cup C_2^m$ the intersection of the right boundary of its cone $C_0$ is closer to $z$ than that of $m$. Thus, it remains to show that $C_0$ of this next vertex contains no edges. We consider the following two cases.
  
\begin{description}
\item [The algorithm follows the edge to $\boldsymbol{v}$.] 
If the algorithm follows the edge to $v$, recall that $v$ lies in $C^w_4$ and $m$ is not the endpoint of a constraint that intersects the interior of triangle $m v w$. In particular, this means that $w$ lies outside of $C_0^v$.  Since $v$ is the closest visible vertex in the subcone of $C_2^m$ that shares a boundary with $C_1^m$, the part of $C_0^v$ below the horizontal line through $m$ must be empty of points visible to $v$ (see Figure~\ref{fig:RoutingToV}a).
  
  \begin{figure}[ht]
    \begin{center}
      \includegraphics{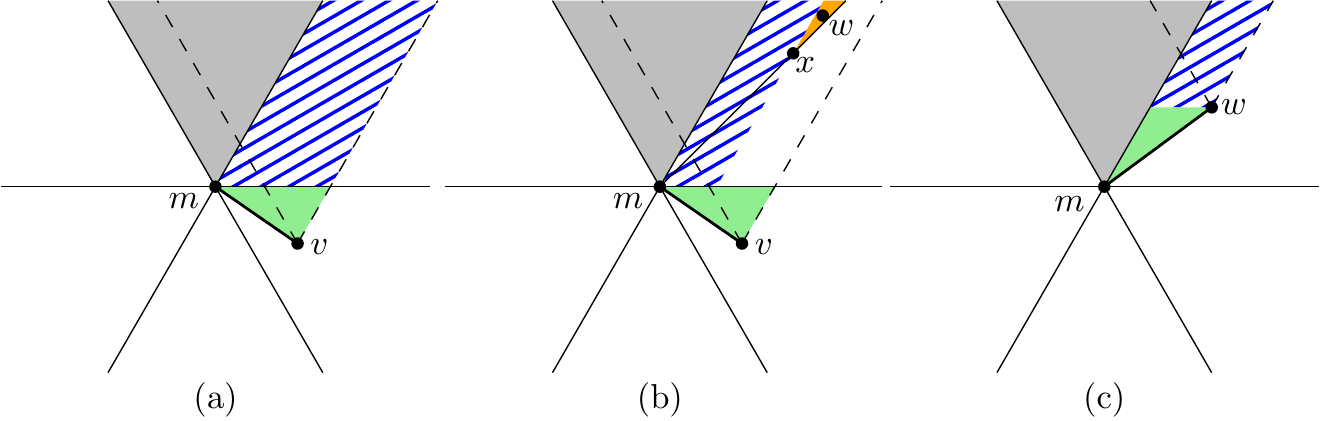}
    \end{center}
    \caption{(a) If $m$ routes to $v$, the union of green and blue regions must be empty of points. (b) An illustration of the proof: if the region is not empty, we find a point $x$ that must have an edge with $m$ that we would have followed instead of $v$. (c) Routing from $m$ to $w$.}
    \label{fig:RoutingToV}
  \end{figure}
  
By the invariant, $C_0^m \cap C_0^v$ is empty of visible points. What remains to be shown is that there are no points visible to $v$ in $C_0^v\setminus C_0^m$ above the horizontal line through $m$. If this region is not empty, we sweep the region using the right boundary of $C_0^m$. Let $x$ be the first vertex hit by this sweep that is visible to $m$ (see Figure~\ref{fig:RoutingToV}b), and consider the canonical triangle $\triangle_{xm}$. Recall that this triangle has $x$ as apex and is bounded by the cone boundaries of the cone of $x$ that contains $m$ and the line through $m$ perpendicular to the bisector of the cone. In particular $\triangle_{xm}$ is empty of points visible to $x$ (since it is contained in the union of $C_0^m$, which is empty), the swept part of $C_1^m$, and a portion of $C_2^m$ that must also be empty by our choice of $v$. This implies that there is an edge from $x$ to $m$. This means that  $w$ must exist. By construction, $\overline{mw}$ forms the smallest angle with the right boundary of $C_0^m$. This means that $x\in \triangle_{mw}$. Furthermore, since $mw$ and $mv$ are visibility edges, Lemma~\ref{lem:ConvexChain} implies the existence of a vertex visible to $w$ in $\triangle_{wm}$. This contradicts the existence of the edge $mw$. Thus, $C_0^v$ is empty of vertices visible to $m$. Suppose that there was a vertex $y$ visible to $v$ in  $C_0^v$, then since $vy$ and $vm$ are visibility edges, Lemma~\ref{lem:ConvexChain} implies the existence of a vertex visible to $m$ in $C_0^v$, which is a contradiction.
  
\item[The algorithm follows the edge to $\boldsymbol{w}$.] As in the previous case, we consider the part below the horizontal line through $w$ and the part above (solid green and dashed blue regions in Figure~\ref{fig:RoutingToV}c, respectively). The former region must be empty or the edge $m w$ would not be present: any point visible to $m$ in this region prevents $m$ from creating an edge to $w$ and vice versa. An argument similar to the one for $v$, showing that the region above the horizontal boundary of $C_1$ is empty, also proves that the region above the horizontal line through $w$ is empty. Thus, $C_0^w$ must be empty of points visible to $w$.
\end{description}\vspace{-2em}
\end{proof}

We now consider the general case, where $u$ may have invalid edges in $C_0$ (see Figure~\ref{fig:ObstacleAvoidance}a). In this case, when $u$ initiates the obstacle avoidance phase, we either reach $z$ or a vertex $m$ that has no edges in $C_1$ and $C_2$ (see Figure~\ref{fig:ObstacleAvoidance}b). This latter case can only occur when $z$ lies in $C_3^m$. Note that this implies that $Q$ intersects both boundaries of $C_1^m$. Therefore, we initiate a new obstacle avoidance phase from $m$ where $C_1$ plays the role of $C_0$. By Lemma \ref{lemma:Reaching the endpoint-1}, the second invocation of the obstacle avoidance phase must reach $z$. 

\begin{figure}[ht]
  \begin{center}
    \includegraphics{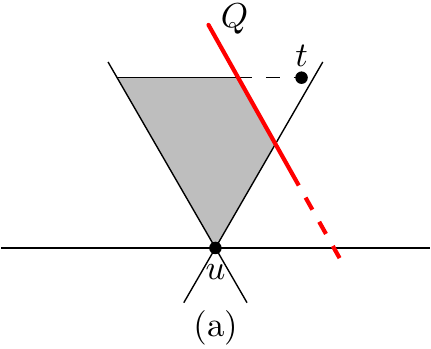}
    \hskip0.5cm
    \includegraphics{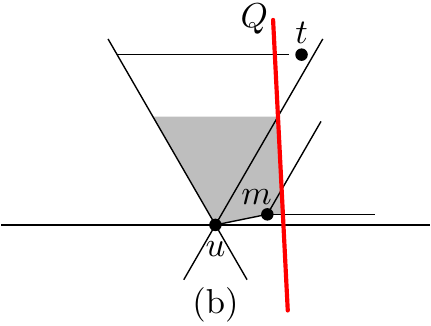}
  \end{center}
  \caption{(a) When $Q$ does not fully block the visibility of $C_0$, we maintain the invariant that the visible portion of the canonical triangle (gray region) must be empty along our routing. Note that edge $uv$ is invalid. (b) The situation where we restart the obstacle avoidance algorithm at $m$.}  \label{fig:ObstacleAvoidance}
\end{figure}

\begin{lemma}\label{lemma:Reaching the endpoint}
When $u$ has no valid edges in the cone containing the destination $t$, the general obstacle avoidance phase initiated by $u$ reaches the right endpoint $z$ of the closest constraint $Q$ blocking visibility to $t$.
\end{lemma}

We note that the above proof relies heavily on the fact that we have exactly 6 cones (and thus we are in the constrained $\Theta_6$-graph). We have a specific example in which the routing strategy described above would fail for 14 cones (for some node, no edge will keep an invariant zone empty, see Figure~\ref{fig:NoVertexExtendingEmptyTriangle}). Thus, a different obstacle avoidance method is needed when the number of cones is not $6$. 

\begin{figure}[ht]
  \begin{center}
    \includegraphics{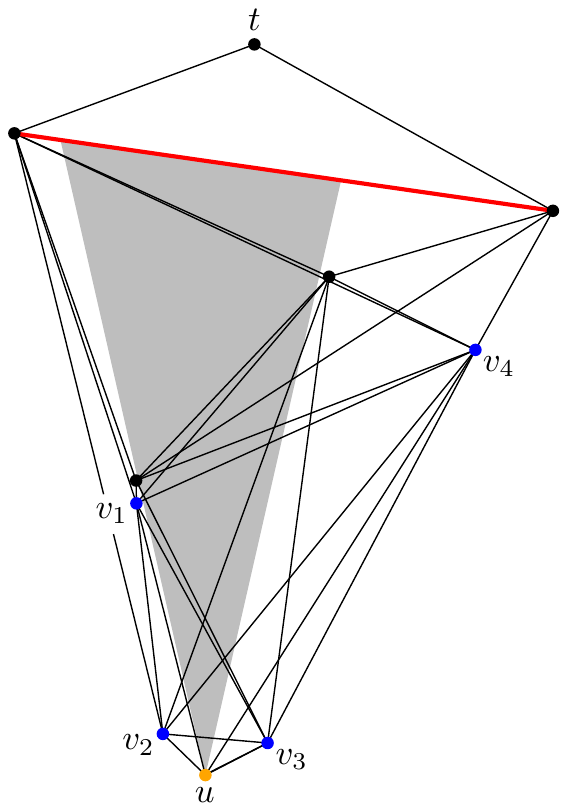}
  \end{center}
  \caption{The situation where no edge from $u$ (orange) to its neighbors $v_1$, $v_2$, $v_3$, and $v_4$ (blue) preserves the empty triangle (gray) used in the above proof, when using 14 cones.}
  \label{fig:NoVertexExtendingEmptyTriangle}
\end{figure}

\subsection{Global Routing Strategy}
\label{sec:GlobalStrategy}
We now have all the pieces in place to describe our routing strategy. Our routing strategy alternates between three phases: while not blocked by an obstacle, we use the classic $\Theta$-routing algorithm. If the current vertex has no valid edges in the cone containing the destination, it must be blocked by a constraint $Q$. In this case, we enter the obstacle avoidance phase to reach the right endpoint of $Q$. Once we reach this endpoint, we check which of the two endpoints of $Q$ is closer to the destination. 

If the closest point to destination is the other endpoint of $Q$, we enter the {\em opposite endpoint phase}, where we route to this other endpoint of $Q$. Note that the two endpoints of $Q$ can see each other, so we can route between them using the strategy introduced in~\cite{BFRV2017RoutingJournal}. The general idea behind the routing strategy in~\cite{BFRV2017RoutingJournal} is to stay as close as possible to the visibility edge between the endpoints of $Q$ in order to avoid becoming stuck behind constraints. 

Once we have reached the endpoint of $Q$ that is closest to the destination, we resume classic $\Theta$-routing. We call this alternation between the three phases the {\em constrained $\Theta_6$-routing strategy}.

\subsection{Convergence}
We now show that our routing algorithm always reaches the destination. First we give a proof of convergence which greatly overestimates the number of steps needed to reach the destination, but it turns out that first showing that the algorithm always reaches the destination simplifies the proof of bounding the number of steps. 

\begin{lemma}
  \label{lem:Convergence}
  The  constrained $\Theta_6$-routing strategy always reaches the destination within a finite number of steps.
\end{lemma}
\begin{proof}
  By construction, each edge followed during the $\Theta$-routing phase gets closer to the destination. Hence, each $\Theta$-routing phase can consist of at most $n$ steps. Similarly, an obstacle avoidance phase performs at most $n$ steps, since each step brings the boundary of cone $C_0$ closer to the endpoint we are routing to.  At the end of an obstacle avoidance phase, we may need an opposite endpoint phase which visits each vertex at most once~\cite{BFRV2017RoutingJournal}. Thus, each cycle of these three phases consists of at most $3n$ steps. 
  
  Thus, in order to show termination it remains to bound the number of alternations between phases. Each invocation of an obstacle avoidance phase is tied to a single constraint $Q$. Even though $Q$ can trigger several obstacle avoidance phases, we claim that the total number is bounded. Let $z$ be the endpoint of $Q$ that is closest to $t$. We claim that between two obstacle avoidance phases triggered by $Q$ we must perform an obstacle avoidance phase using another constraint $Q'$ whose endpoint $z'$ that is closest to $t$, lies further away from $t$ than $z$. 
  
Let $D$ be the closed disk with center $t$ and radius $|t z|$. We need to show that before using $Q$ for another obstacle avoidance phase, we must reach an endpoint $z'$ that lies outside $D$. 
  
  In order for $Q$ to trigger another obstacle avoidance phase, the routing path needs to first reach a vertex $v$ such that $Q$ blocks visibility to $t$ from $v$. This implies that $v$ and $t$ lie in different halfplanes with respect to the line through $Q$. Furthermore, $v$ cannot lie on this line, since $Q$ needs to block visibility between $v$ and $t$. Let $H_v$ be the open halfplane that contains $v$ and let $H_t$ be its complementary closed halfplane (see Figure~\ref{fig:VisitFurtherConstraint}). 

\begin{figure}[ht]
  \begin{center}
  \includegraphics{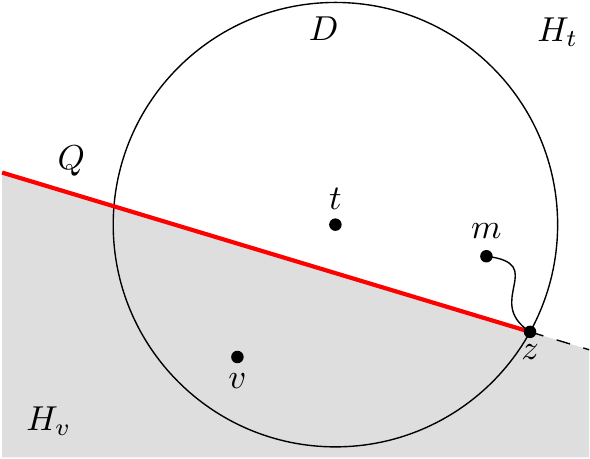}
  \end{center}
  \caption{For constraint $Q$ to trigger another obstacle avoidance phase after reaching $z$, we need to reach to a vertex $v \in H_v$ (gray, note that $v$ need not lie in $D$). However, after leaving $z$, we first reach a vertex $m \in H_t \cap D$. Since $Q$ prevents the path between $m$ and $v$ from remaining inside $D$, we need to leave $D$ at some point and this can only be caused by a new constraint $Q'$ whose endpoints are further from $t$ than $z$ is.}
  \label{fig:VisitFurtherConstraint}
\end{figure}

Consider the routing step we performed at $z$ after reaching it as the endpoint of the obstacle avoidance phase of $Q$. Specifically, we look at which of the possible phases this step was executed. If the algorithm started a $\Theta$-routing phase, the very first step must be towards the interior of $H_t \cap D$ (since each step of $\Theta$-routing gets closer to $t$ and follows an edge in the subcone that contains $t$). The other option is that we immediately start another obstacle avoidance phase because of a new constraint. If the endpoint closest to $t$ of this new constraint lies outside $D$ we are done (since the obstacle avoidance phase will take us to the endpoint $z'$ that we claimed), so assume that this endpoint lies inside $D$. Furthermore, since this constraint blocks visibility from $z$ to $t$, its endpoint closest to $t$ must lie in $H_t$. Thus, we conclude that regardless of which strategy we used to route, after we leave $z$ we must reach an intermediate vertex $m$ that lies in $H_t \cap D$ before we visit $v$. 
  
  Overall, we have that our proposed routing reaches $z$, then $m \in H_t \cap D$, and somehow then goes to some vertex $v \in H_v$. Recall that $z$ was defined as the closest endpoint of $Q$ to $t$. In particular, the other endpoint is outside $D$. Thus, even though $v$ may be in $D$, in order for the routing strategy to reach any point of $H_v$ it must leave $D$ at some point. Since during the $\Theta$-routing phase we cannot get further away from $t$, the only way that this can happen is via an obstacle avoidance phase where both endpoints lie outside $D$, proving our claim. 
  
  With this claim shown, we can now proceed to bound the number of alternations between the different phases. Let $Q_1, \ldots, Q_k$ be all the constraints sorted by decreasing distance of their closest endpoint to $t$. Let $z_i$ be the endpoint of $Q_i$ closest to $t$ (i.e., for any $1\leq i<j \leq k$, we have that $|tz_i|>|tz_j|$). Notice that $Q_1$ cannot invoke more than one obstacle avoidance phase since there are no constraints whose closest endpoint $z_i$ is further from $t$ than $z_1$. By a similar reasoning, we can show that $Q_2$ can trigger two obstacle avoidance phases, $Q_3$ four such phases, and in general $Q_i$ cannot invoke an obstacle avoidance phase more than $2^{i-1}$ times. Since there are $k$ constraints, there cannot be more than $2^0+2^1+ \ldots 2^{k-1}= 2^k - 1$ invocations of an obstacle avoidance phase. As argued at the beginning of the proof, we execute at most $3n$ steps between two obstacle avoidance phases, thus the total number of steps is upper bounded by $O(n \cdot 2^k)$. 
\end{proof}

Note that the above reasoning shows that a single constraint can trigger an obstacle avoidance phase many times. Having shown that our algorithm terminates after a finite number of steps, we can refine the argument to reduce the number of triggers per constraint to exactly one.

\begin{lemma}\label{cor_visitonce}
Let $Q$ be a constraint and let $z$ be the endpoint of $Q$ that is closest to $t$. Vertex $z$ can be visited as the final vertex of at most one obstacle avoidance or opposite endpoint phase. 
\end{lemma}
\begin{proof}
When we reach $z$ at the end of an obstacle avoidance or opposite endpoint phase, we execute a step in the $\Theta_6$-routing strategy. Since this strategy is memoryless\footnote{An algorithm is called \emph{memoryless} if it makes the same decision when presented with the same input, i.e., it does not store previous decisions.}, the routing strategy follows the same edge from $z$ every time we reach it. This implies that $z$ cannot be visited twice using an obstacle avoidance or opposite endpoint phase, since otherwise the path would cycle indefinitely, contradicting Lemma~\ref{lem:Convergence}.
\end{proof}

This immediately gives a linear bound on the number of phase changes, implying a quadratic bound on the number of steps. We now use a more detailed analysis of the circumstances in which a vertex may be visited to tighten this further to $O(n)$. In order to do this, we first determine the number of constraints that can fully block visibility in a cone of a vertex $v$. 

\begin{lemma}
\label{lem:ThreeConstraints}
For any vertex $v$ there can be at most three constraints that fully block visibility in some cone of $v$. 
\end{lemma}

\begin{proof}
  There are six cones around $v$ and each one can only be fully blocked by at most one constraint. This already implies that at most six constraints can fully block visibility in some cone of $v$. In the following we reduce the number to three.

\begin{figure}[ht]
  \begin{center}
  \includegraphics{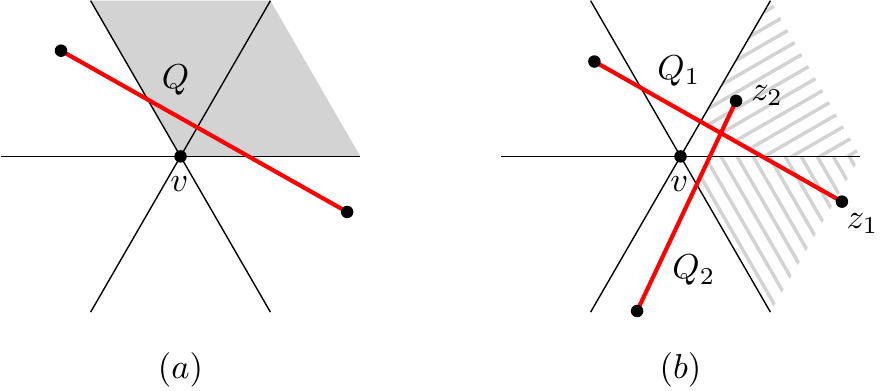}
  \end{center}
  \caption{(a) Constraint $Q$ fully blocks visibility in cone $C_0$ and $C_1$ of vertex $v$. Notice that its endpoints lie in $C_0$ and $C_2$. (b) There cannot be two adjacent cones fully blocked by different constraints: in order for this to happen, both endpoints should lie behind the other constraint (dashed regions) and this is only possible if the constraints cross.}
  \label{fig:FullyBlocking}
\end{figure}

  We first observe that if a constraint $Q$ fully blocks visibility in some cone(s), the endpoints of $Q$ must be in the cones adjacent to those blocked by it (say, if $Q$ blocks cones $C_0$ and $C_1$, then the endpoints of $Q$ must lie in $C_5$ and $C_2$, see Figure~\ref{fig:FullyBlocking}a). We use this fact to show that there cannot be two adjacent cones that are fully blocked by different constraints: assume, for the sake of contradiction, that we can have two constraints $Q_1$ and $Q_2$ fully blocking cones $C_1$ and $C_2$, respectively. Let $z_1$ be the endpoint of $Q_1$ that lies in $C_2$ and let $z_2$ be the endpoint of $Q_2$ that lies in $C_1$. Because both $C_1$ and $C_2$ are fully blocked, neither $z_1$ nor $z_2$ are visible from $v$. Thus, they must both lie behind the constraint blocking visibility in their cones. However, this would imply that $Q_1$ and $Q_2$ cross each other (see Figure~\ref{fig:FullyBlocking}b), contradicting that the set of constraints is plane. Naturally, the same argument applies to any other two adjacent cones, hence we conclude that between two different constraints that fully block visibility there is at least one cone with visible vertices. Since we have 6 cones in total, the limit of three constraints follows.
\end{proof}

\begin{lemma}
\label{lem:LinearSteps}
  The  constrained $\Theta_6$-routing strategy always reaches the destination in $O(n)$ steps. 
\end{lemma}
\begin{proof}
Consider any vertex $v$ and consider how we reached it. 

\begin{description}
\item[1) $\boldsymbol{v}$ is reached during a $\boldsymbol{\Theta}$-routing phase.] Since the routing strategy in this phase is memoryless, we would make the same routing step from $v$ every time we reach it. In particular, this would imply that $v$ cannot be visited twice using a $\Theta$-routing phase (otherwise, the path would cycle indefinitely, contradicting with Lemma~\ref{lem:Convergence}). Hence, we conclude that $v$ is visited once during a $\Theta$-routing phase during the whole routing algorithm.

\item[2) $\boldsymbol{v}$ is reached during an avoidance phase of constraint $\boldsymbol{Q}$.] We consider two subcases:

\begin{description}
\item[2.1) $\boldsymbol{v}$ is not an endpoint of $\boldsymbol{Q}$.] Let $u$ be the vertex that initiated the avoidance phase and first consider the case in which $Q$ completely blocks visibility of $u$ in the cone containing $t$ (see Figure~\ref{fig:ThetaRoutingLongStep}b). In this situation, the same cone remains empty for all vertices along the path (including $v$). By Lemma~\ref{lem:ThreeConstraints}, if $v$ is visited more than three times as part of an obstacle avoidance path, two of them share the same cone. Both of these times, the obstacle avoidance and opposite endpoint phases would end up at $z$, the endpoint of $Q$ closest to $t$, contradicting Lemma~\ref{cor_visitonce}. Thus, we conclude that $v$ can be reached this way at most three times.

It is possible that $Q$ did not block the visibility in the cone completely (i.e., we initiated the obstacle avoidance phase because the edge was invalid, see Figure~\ref{fig:ThetaRoutingLongStep}a). This situation is very similar to the case in which visibility was completely blocked. The only difference is that the choice of the edge we follow at $v$ depends on the cone that contained $t$ when we started this obstacle avoidance phase as well as on whether or not $v$ has edges in the two adjacent cones. We again conclude that if $v$ is visited more than a constant number of times in this way, the algorithm would route to the same neighbour of $v$, eventually ending at the same endpoint of $Q$ and contradicting Lemma~\ref{cor_visitonce}. 

\item[2.2) $\boldsymbol{v}$ is an endpoint of $\boldsymbol{Q}$.] As argued in Lemma~\ref{cor_visitonce}, $v$ can only be visited once during the whole execution of the algorithm if it is the endpoint that is closest to $t$. Similarly, if $v$ is the endpoint that is furthest away from $t$, we know the algorithm enters the opposite endpoint phase and routes to the opposite endpoint of $Q$. Note that $v$ could be visited several times this way (see Figure~\ref{fig:ObstacleManyVisits}). However, notice that $v$ can never be visited twice because of the same constraint $Q$, as this would imply that we visit the same closest endpoint twice as well, contradicting Lemma~\ref{cor_visitonce}. Thus, during the entire execution of the algorithm, we can visit at most $3n-6$ vertices as the endpoint of a constraint that is not closest to $t$, since the set of constraints is plane. 
\end{description}

\begin{figure}[ht]
  \begin{center}
  \includegraphics{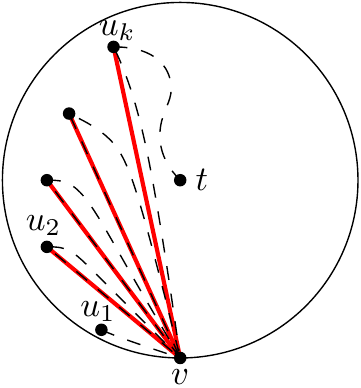}
  \end{center}
  \caption{A vertex $v$ can be visited $\Omega(n)$ times as the endpoint not closest to $t$. This implies that $v$ is the endpoint of many constraints and in all of them it is further away from $t$ than the other endpoint $u_2$, ..., $u_k$. For clarity, the disk centred at $t$ passing through $v$ is drawn (as solid black), and a possible routing path that visits $v$ multiple times is also shown (in dashed black).}
  \label{fig:ObstacleManyVisits}
\end{figure}

\item[3) $\boldsymbol{v}$ is reached during an opposite endpoint phase.] Every time a vertex is part of a path in the opposite endpoint phase, Lemma 3 of \cite{BFRV2017RoutingJournal} shows that at least one of its cones is empty. 
\end{description}

Hence, excluding case 2.2, each vertex is visited a constant number times. Since case 2.2 adds at most $3n-6$ visited vertices during the entire execution of the algorithm, this implies that a total of $O(n)$ steps are executed as claimed. 
\end{proof}

\begin{theorem}
\label{theo:RoutingTheta6}
  There exists a 1-local $O(1)$-memory routing algorithm for the constrained $\Theta_6$-graph that reaches the destination in $O(n)$ steps. 
\end{theorem}
\begin{proof}
  The algorithm is 1-local by construction, since we consider only information about vertices the current vertex is connected to. The $\Theta$-routing phase does not require any memory. The obstacle avoidance phase and opposite endpoint phase store a single vertex each and this information is discarded when the phase ends. Hence, the algorithm requires $O(1)$ memory. Lemma~\ref{lem:LinearSteps} shows that the algorithm terminates in $O(n)$ steps. 
\end{proof}

\section{Routing on the Visibility Graph}
We now return our attention to our main goal: routing on the visibility graph. Since in the previous section we presented a routing algorithm for the constrained $\Theta_6$-graph, we first show that we can use this algorithm to route on the visibility graph as well. Afterwards, we also describe how to modify the constrained $\Theta_6$-routing algorithm to route on the visibility graph directly without locally determining the edges of the constrained $\Theta_6$-graph. 

We note that, unfortunately, the length of the paths resulting from these two approaches need not be related to the length of the shortest path in the visibility graph. Since we cannot determine locally which endpoint of a constraint is closest to $t$, the routing algorithms may follow a path to an endpoint arbitrarily far away, preventing us from being competitive. However, we show in Section~\ref{sec:LowerBounds} that no deterministic local routing strategy can be $o(\sqrt{n})$-competitive with respect to the length of the shortest path. 

\subsection{Using the Constrained $\boldsymbol{\Theta_6}$-Graph}
In order to use the constrained $\Theta_6$-routing algorithm from the previous section, we need to determine locally at a vertex which of its visibility edges are part of the constrained $\Theta_6$-graph. Since it is easy to locally determine at a vertex $u$ if a vertex $v$ is the closest vertex in one of its subcones, we focus on the situation where this is not the case and we thus have to determine at $u$ if it is the closest vertex in one of the subcones of $v$. Let the \emph{constrained canonical triangle} of $v$ be $\triangle_{vu}$ clipped using the constraints intersecting the boundary of the canonical triangle with one endpoint at $u$ (see Figure~\ref{fig:ConstrainedCanonicalTriangle}). Note that we can determine the constrained canonical triangle of $v$ locally at $u$. 

\begin{figure}[ht]
  \begin{center}
    \includegraphics{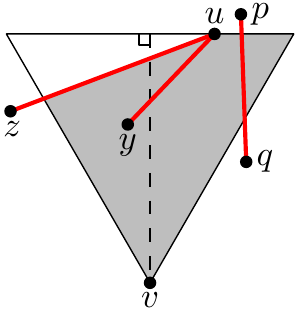}
  \end{center}
  \caption{The constrained canonical triangle of $v$ (gray). Constraint $u z$ is used to clip the triangle. Constraint $u y$ does not clip the triangle, since it does not cross the triangle boundary. Constraint $p q$ does not clip the triangle, since it has no endpoint at $u$.}
  \label{fig:ConstrainedCanonicalTriangle}
\end{figure}

\begin{lemma}
  \label{lem:Theta6Edges}
  Let $u$ and $v$ be two vertices such that $v$ is not the closest vertex to $u$ in any subcone of $u$. Edge $u v$ is part of the constrained $\Theta_6$-graph if and only if $u$ does not have any visible vertices in the constrained canonical triangle of $v$. 
\end{lemma}
\begin{proof}
  We first note that we can consider the subcone of $v$ that contains $u$ to be the full cone: If the constraint defining the subcone ends in the constrained canonical triangle, Lemma~\ref{lem:ConvexChain} implies that it also contains a vertex visible to $u$, correctly implying that $u v$ is not an edge. If the constraint does not end in the constrained canonical triangle, the part of the constrained canonical triangle outside the subcone is not visible to $u$ and hence it does not influence the decision at $u$. 

  It is easy to see that if $u$ has any visible vertices in the constrained canonical triangle of $v$, $u v$ is not an edge of the constrained $\Theta_6$-graph: Consider the vertex $x$ such that the smaller angle of $u x$ and $u v$ is minimized. Since the angle is minimized, $u$ is not the endpoint of any constraints intersecting triangle $uvx$, so we can apply Lemma~\ref{lem:ConvexChain} to $uvx$. This gives us a vertex inside the constrained canonical triangle that is visible to $v$. Hence, $u$ is not the closest visible vertex to $v$ and thus $u v$ is not an edge of the constrained $\Theta_6$-graph. 
  
  Next we show that if $u$ has no visible vertices in the constrained canonical triangle of $v$, $u v$ is an edge of the constrained $\Theta_6$-graph. We prove this by contradiction, so assume that $u v$ is not an edge of the constrained $\Theta_6$-graph. This implies that there exists a vertex $x$ in the subcone of $v$ that contains $u$ that is closer to $v$ than $u$ is. Hence, $x$ lies in the constrained canonical triangle. Applying Lemma~\ref{lem:ConvexChain} to $u v x$ gives us a vertex inside the constrained canonical triangle that is visible to $u$, contradicting that $u$ has no visible vertices in this region.  
\end{proof}

\subsection{Routing Directly on the Visibility Graph}
In order to route directly on the visibility graph, instead of at each vertex computing the local neighbourhood in the constrained $\Theta_6$-graph, the constrained $\Theta_6$-routing algorithm needs to be modified. We do this in such a way that the vertices do not need to store any fixed cone orientations.

When a vertex $s$ wants to send a message, it picks an arbitrary cone orientation and stores it in the message it sends. We note that a vertex can pick a different orientation of the cones for each message that it sends and this only requires a constant amount of storage. Since the orientation is stored in the message, vertices do not need to agree on a fixed orientation in advance, as every vertex along the routing path can extract the orientation from the message and use that for its decisions. 

Like in the constrained $\Theta_6$-routing algorithm, routing directly on the visibility graph works in three phases: $\Theta$-routing, obstacle avoidance, and opposite endpoint. During the $\Theta$-routing phase a vertex $u$ simply sends the message to the closest vertex in the cone that contains $t$, again limiting the edges it is allowed to follow to the edges that end in $\triangle_{ut}$. 

During the obstacle avoidance phase, we start by routing to either endpoint of the constraint blocking visibility to $t$. Since we are routing on the visibility graph, Lemma~\ref{lem:ConvexChain} tells us that there is a convex chain of visibility edges to these endpoints. Hence, in order to reach an endpoint of the constraint, we follow one of these convex chains. In order to determine the next edge on the chain at an intermediate vertex $m$, the message needs to store the predecessor of $m$ on the chain and whether the path should continue to the next clockwise or counter-clockwise edge of $m$. The next edge along the convex chain at $m$ is the edge that minimizes the angle with the line through $m$ and the predecessor of $m$ in the stored direction. 

When we arrive at an endpoint of a constraint, we can determine the location of the other endpoint, since they are connected in the visibility graph. Using this information, we can determine if this constraint is the one that caused the obstacle avoidance phase by checking if it blocks visibility of $u$ to $t$. If this is the case, we also determine which of the two endpoints is closer to $t$. If we are not yet at the endpoint closest to $t$, we start the opposite endpoint phase, which is now simplified to following the edge in the visibility graph to the other endpoint of the constraint. 

\begin{theorem}
\label{theo:RoutingVisibilityGraph}
  There exists a 1-local $O(1)$-memory routing algorithm for the visibility graph that reaches the destination in $O(n)$ steps. 
\end{theorem}
\begin{proof}
  We first note that locality follows from the fact that we only need to consider the neighbours of the current vertex in each of the steps. The memory bound follows from the fact that we need to store only the orientation of the cones in the message, as well as the starting vertex of the obstacle avoidance phase and the previous vertex along the obstacle avoidance path. 
  
  It remains to bound the number of steps. This algorithm has properties similar to those of the constrained $\Theta_6$-routing algorithm. First, the $\Theta$-routing phase always gets closer to the destination and thus cannot repeat vertices. This implies that Lemma~\ref{lem:Convergence} also holds for this routing algorithm. This in turn implies that a vertex can be the closest endpoint of an obstacle avoidance or opposite endpoint phase at most once. Next, since the obstacle avoidance path is convex, this implies that this path visits a subset of the vertices visited by the obstacle avoidance phase of the constrained $\Theta_6$-routing algorithm. Finally, the opposite endpoint phase consists of at most a single edge, hence this phase too is a subpath of its constrained $\Theta_6$-routing counterpart. Hence, when we compare the path of this routing algorithm to the constrained $\Theta_6$-routing path that uses the same cone orientation, the routing path on the visibility graph is a subpath of the constrained $\Theta_6$-routing path. Hence, it takes at most $O(n)$ steps. 
\end{proof}

\section{Lower Bounds}
\label{sec:LowerBounds}
In this section we provide a number of lower bounds on the competitive ratio of any deterministic local routing algorithm on the visibility graph compared to the shortest path in that graph. We give bounds both on the total length of the path, and on the number of steps taken. Our first two lower bounds are adaptations of the lower bound given by Bose~\etal~\cite{BFRV2017RoutingJournal} in Theorem~1. Their focus of interest is the constrained $\Theta_6$-graph, where they showed that no deterministic 1-local routing algorithm on this graph can be $o(\sqrt n)$-competitive. In this section, we modify their construction for the visibility graph instead.

We start by giving an overview of their bound. For a given $n$, the general idea is to initially construct a plane graph with a number of vertices quadratic in $n$. This graph is grid shaped, with the source at the bottom and sink at the top (see Figure~\ref{fig:RoutingLowerBoundGrid}). We run the routing algorithm and see which path it follows on the large graph. Note that the graph is of quadratic size, but the shortest path uses a linear number of edges. Thus, competitive algorithms cannot afford to fully explore the graph. Once we know the path of the routing strategy, we trim the graph to one of linear size in $n$, removing the portions of the graph that the routing algorithm did not explore (and in the process we insert a shorter path, see Figure~\ref{fig:RoutingLowerBoundGraph}). The key point of the construction is that the routing algorithm has the exact same information in both the original and the trimmed graph. Since the algorithm is deterministic, it will make the same choices and follow the same path in both cases.

\begin{figure}[h]
  \begin{center}
    \includegraphics{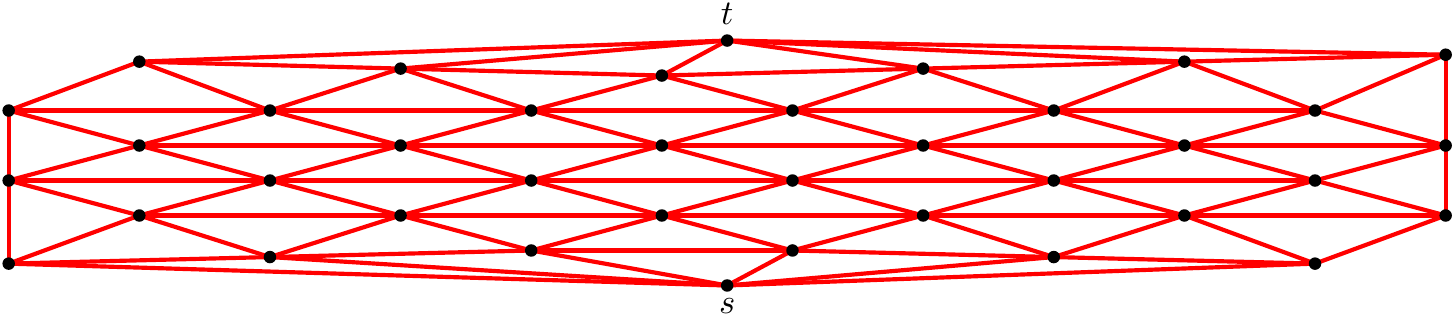}
  \end{center}
  \caption{The initial graph for the lower bound construction. The edges shown are the ones that would be created in the $\Theta_6$-graph. The same graph can be realized as a visibility graph by simply turning every edge into a constraint.}
  \label{fig:RoutingLowerBoundGrid}
\end{figure}

We start by showing a lower bound on the competitive ratio with respect to the {\em number of steps} in the path (often referred to as the {\em hop distance}) as opposed
to the length of the shortest path. 

\begin{lemma}
\label{lem:LinearStepsLowerBound}
  No deterministic 1-local routing algorithm is $o(n)$-competitive with respect to the number of steps of the shortest path on all pairs of vertices of the visibility graph on $n$ vertices, regardless of the amount of memory it is allowed to use. 
\end{lemma}
\begin{proof}
We need to show that the large and the trimmed graphs can both be realized as a visibility graph (instead of the constrained $\Theta_6$-graph). The large graph consists of $n\times n$ vertices in the unit grid in which the vertices in every other column are shifted half a unit in the $x$-coordinate, and then we scale the instance in the $x$-coordinate by a factor of $n$ (see exact details in~\cite{BFRV2017RoutingJournal}). Since the graph is a maximal plane graph (i.e., any additional edge would create a crossing), we can realize it as a visibility graph by making every edge a constraint. 

Now we explain how to also make the trimmed graph. Consider any deterministic 1-local routing algorithm and the path it follows from $s$ to $t$. Let $\pi$ be the path consisting of the first $n/2$ steps of this routing path. We note that since the initial graph has $n$ rows, the shortest path must have at least $n/2$ steps (by following edges on the left or right boundary you can skip one every two rows) and thus $\pi$ is well defined. 
  
  Next, we modify the initial graph in exactly the same way as done by Bose~\etal: for every vertex visited by $\pi$, we mark that vertex, and all of its neighbours. We also mark $t$, its neighbors, and all constraints whose two endpoints are marked vertices. The trimmed graph consists of the visibility graph consisting of marked vertices and constraints only. That is, we remove any non-marked vertex, edge, and constraint, and ``update'' the visibility graph (see Figure~\ref{fig:RoutingLowerBoundGraph}). 
  
\begin{figure}[h]
  \begin{center}
    \includegraphics{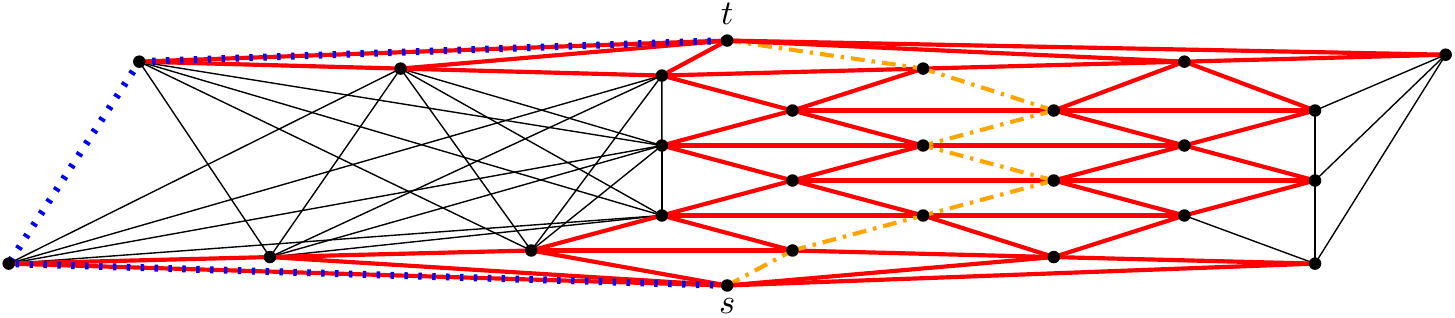}
  \end{center}
  \caption{Once we know the path $\pi$ (orange and dash-dotted) followed by the routing strategy we can create the trimmed graph: this graph will have the same local information as in the initial graph, hence the same choices will be made. However, the trimmed graph contains a shorter path with a constant number of steps (blue and dotted). Constraints are shown in thick red. Note that in the left side of the graph there could be quadratically many edges (in solid black), but this can be reduced to a linear number by choosing any plane graph and adding those edges as constraints.}
  \label{fig:RoutingLowerBoundGraph}
\end{figure}
  
    Now we analyze the performance of the routing algorithm in the trimmed graph. Consider the $n$ columns of the construction. We say that a vertex \emph{touches} a column if it lies in that column or has a neighbour in that column. Observe that every vertex of $\pi$ (other than $s$ and $t$) touches a constant number of columns. Since $\pi$ consists of $n/2$ steps, there exists a column $c$ that is touched at most $O((n/2)/n) = O(1)$ times. 
  
  Regardless of $\pi$, we observe that there is a path from $s$ to $t$ that consists of $O(1)$ steps: follow the edge from $s$ to the column $c$, and go upwards in that column until you reach the top row, and follow the edge to $t$. There are only $O(1)$ vertices in that column (i.e., the touched vertices). Between any two consecutive vertices we need a constant number of steps. Since any deterministic 1-local routing algorithm would follow a path consisting of at least $n/2$ steps, we get the claimed lower bound on the competitive ratio with respect to the number of steps. 
  
  In order to complete the proof, we need to argue that the trimmed graph has linear size. Since $\pi$ consists of $n/2$ steps and every vertex other than $s$ and $t$ has a constant number of neighbors, we obtain a graph with $\Theta(n)$ vertices. However, note that the resulting graph could have quadratically many edges (see the left side of Figure~\ref{fig:RoutingLowerBoundGraph} for example). We can reduce the number of edges to linear by adding additional constraints. The only requirement is that these constraints do not cross the edges used in the short path in column $c$. For example, we can add an arbitrary plane graph where every edge is a constraint. 
\end{proof}

The same construction can be used to show a lower bound on the competitiveness in terms of the length of the shortest path.

\begin{lemma}
  \label{lem:sqrtnLowerBound}
  No deterministic 1-local routing algorithm is $o(\sqrt{n})$-competitive with respect to the length of the shortest path on all pairs of vertices of the visibility graph on $n$ vertices, regardless of the amount of memory it is allowed to use. 
\end{lemma}
\begin{proof}
  The construction of both the large and the trimmed graphs is the same as in Lemma~\ref{lem:LinearStepsLowerBound}. The main change is that now the analysis is based on Euclidean length instead of number of steps. Recall that we scaled the instance by a factor of $n$ in the $x$-coordinates. This in particular implies that any non-vertical edge of the graph has length at least $n$. Consider the path $\pi$ consisting of the first $n/2$ steps taken by the routing algorithm. If none of these edges is vertical, we obtain a lower bound of $n^2/2$ on the length of the routing path. If at least one edge is vertical, we observe that both of these vertical sides are at distance $n^2/2$ from $s$, thus giving the same lower bound on the length of the path. 
  
  However, in the trimmed path we can find a shorter path. Consider the $2 \sqrt{n}$ columns\footnote{Since we need to consider the Euclidean distance traveled to reach these columns later, we cannot consider all $n$ columns.} at distance at most $n \sqrt{n}$ from $s$. Since $\pi$ consists of $n/2$ vertices, there exists a column $c$ that is touched at most $\sqrt{n}$ times. 
  
  For ease of exposition, first consider the case in which $c$ is touched only by $s$ and $t$. This would give us a path of length at most $O(n \sqrt{n})$: go from $s$ to the column $c$ and follow the upward edges to $t$. Each of the horizontal steps has length $O(n \sqrt{n})$ and there are two of them in total, whereas the $n$ upward edges have unit length. 
  
  In the general case, recall that $c$ is touched at most $\sqrt{n}$ times. For each vertex that is touched, we need to make a small detour via the next column, adding an extra cost of $O(n)$ per detour, leading to a total path length of $O(n \sqrt{n})$. Hence, the ratio between the two path lengths tends to $\omega(\sqrt n)$, giving the lower bound. 
\end{proof}

Next, we show that this lower bound can be improved in some cases. A common strategy for routing is to use the segment $st$ as a guide to reach the destination. Thus, even though one is allowed to use all edges of the graph, routing algorithms often only consider edges of the subgraph induced by the endpoints of edges that cross $st$. For any such algorithm we can increase the lower bound to show that no $o(n)$-competitive algorithm exists in terms of the shortest path distance.

\begin{lemma}
\label{lem:LinearLowerBound}
  No deterministic 1-local routing algorithm, that considers only edges of the subgraph induced by the endpoints of edges that cross $st$, is $o(n)$-competitive with respect to the length of the shortest path on all pairs of vertices of the visibility graph on $n$ vertices, regardless of the amount of memory it is allowed to use. 
\end{lemma}
\begin{proof}
In order to obtain this lower bound, we modify the lower bound for routing on the constrained $\Theta_6$-graph, presented by the authors~\cite{BKRV2017Routing} in Lemma~4.1. The proof of this claim only needs one graph (shown in Figure~\ref{fig:LowerBound}). 

\begin{figure}[th]
  \begin{center}
    \includegraphics{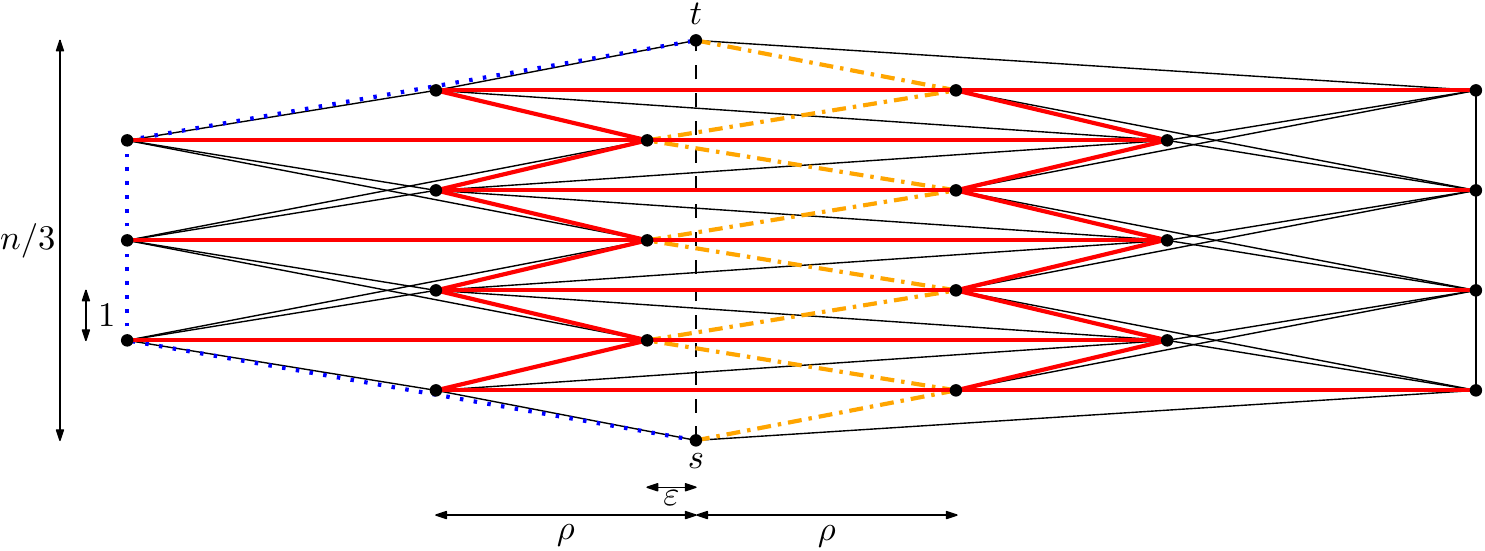}
  \end{center}
  \caption{Lower bound construction: the shortest path in the subgraph induced by all endpoints of edges crossing $st$ (orange and dash-dotted) is about $n/12$ times as long as the shortest path in the visibility graph (blue and dotted). Constraints are shown in thick red and the remaining edges are shown in solid black.}
  \label{fig:LowerBound}
\end{figure}

The construction is as follows (see Figure~\ref{fig:RoutingLowerBoundConstruction}): start with $3$ columns of $n/3$ vertices each, aligned on a grid\footnote{For simplicity, we assumed that $n$ is a multiple of 3. This assumption can be removed by placing the 1 or 2 remaining points far enough away from the remainder of the point set.}. We add a constraint between every horizontal pair of vertices of two consecutive columns. We also add constraints from every vertex that is in an odd row of the first two columns to the vertex in the next column that is in either the next or previous row. Next, we shift every odd row by $1/2+\varepsilon$ units to the right (for some arbitrarily small but positive $\varepsilon<1/2$). We add a vertex $s$ below the lowest row and a vertex $t$ above the highest row, centered the first two vertices on said row. Finally, we stretch the point set by a factor $2\rho$ in the horizontal direction, for some large constant $\rho$. When we construct the visibility graph on this point set, we get the graph shown in Figure~\ref{fig:LowerBound}.

\begin{figure}[h]
  \begin{center}
    \includegraphics{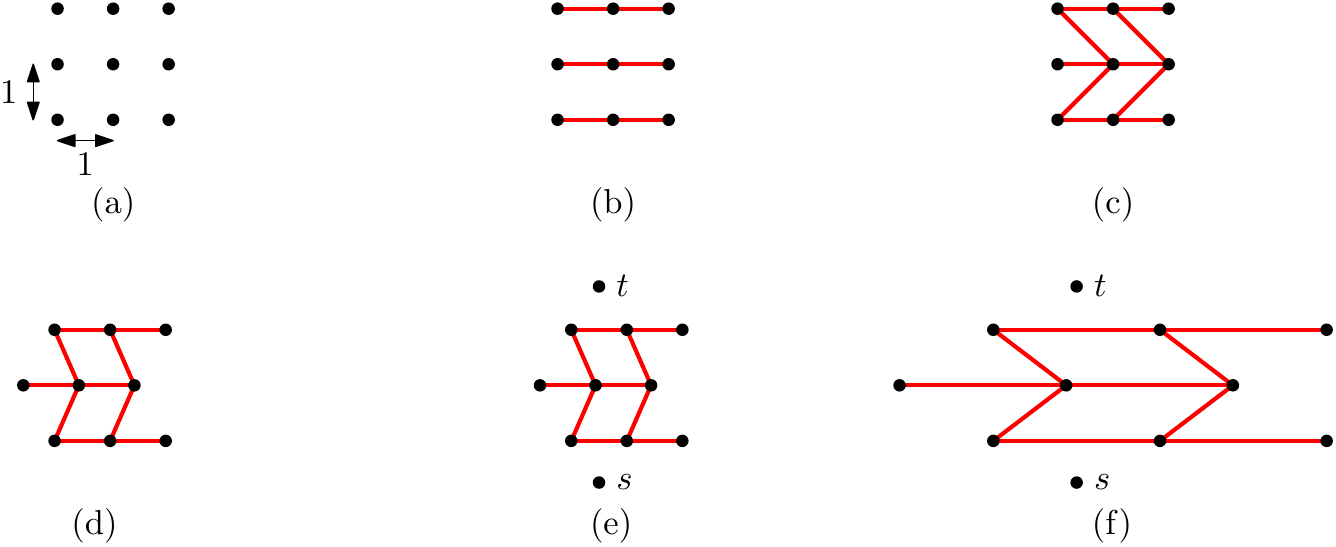}
  \end{center}
  \caption{Constructing the lower bound: (a) the initial point set, (b) adding the horizontal constraints, (c) adding the constraints between rows, (d) shifting the rows, (e) adding $s$ and $t$, (f) stretching the construction.}
  \label{fig:RoutingLowerBoundConstruction}
\end{figure}

The shortest path in the visibility graph goes from $s$ to the leftmost column in one step, travels vertically upwards, and goes to $t$ in one step. Ignoring the terms that depend on $\varepsilon$, the first and last step have length less than $2\rho+2$ whereas the vertical steps each have unit length, giving a path whose total length is less than $4\rho+n/3+4$. 

However, if we restrict ourselves to the subgraph induced by the endpoints of edges that cross $st$ the path becomes significantly longer: we must now zig-zag left and right in a path of $n/3$ steps, each time crossing the segment $st$ (see Figure~\ref{fig:LowerBound}). Ignoring the terms that depend on $\varepsilon$, each edge of this path has length at least $\rho$, giving an overall lower bound of $\rho \cdot n/3$ for the length of any restricted path. Hence, the ratio between the two bounds approaches $n/12$, since $\lim_{\rho \rightarrow \infty} \frac{\rho \cdot \frac{n}{3}}{4 \rho + \frac{n}{3} + 4} = \frac{n}{12} \,.$ 
  
  Since the shortest path in the subgraph is $n/12$ times the length of the shortest path in the visibility graph, no routing algorithm that considers only the subgraph can be $o(n)$-competitive with respect to the length of the shortest path in the visibility graph. 
\end{proof}

\section{Conclusion}
We presented the first deterministic 1-local $O(1)$-memory routing algorithms for the visibility graph that does not require the computation of a planar subgraph. Unfortunately, our algorithms do not give any guarantees on the length of the routing path, only on the number of edges used. A natural improvement would be the design of a routing strategy that is competitive with respect to the length of the shortest path. 

Our lower bounds show that $o(\sqrt{n})$-competitiveness is not possible (and that it will be even hard to obtain $o(n)$-competitiveness). The same lower bounds also give rise to the following questions: Can we design an $O(\sqrt{n})$-competitive deterministic 1-local routing strategy? Can we actually beat these lower bounds by introducing randomness into the routing algorithms?  

\subsection*{Acknowledgements}
Part of this work was performed at the Sendai Workshop on Discrete and Computational Geometry and the Shonan Meeting 106 - Geometric Graphs: Theory and Applications. We thank the participants of both workshops for providing a fun and stimulating research environment. 
\bibliography{references}
\end{document}